\documentclass[%
 aip,
 jmp,
 amsmath,amssymb,
 preprint,%
]{revtex4-1}

\usepackage{graphicx}
\usepackage{dcolumn}
\usepackage{bm}

\usepackage[utf8]{inputenc}
\usepackage[T1]{fontenc}
\usepackage{etoolbox}

\usepackage{amsthm,mathtools}
\usepackage{enumitem}
\usepackage{mathrsfs}
\usepackage{microtype}
\usepackage[colorlinks=true,linkcolor=blue,citecolor=blue,urlcolor=blue]{hyperref}

\newtheorem{theorem}{Theorem}[section]
\newtheorem{proposition}[theorem]{Proposition}
\newtheorem{lemma}[theorem]{Lemma}
\newtheorem{corollary}[theorem]{Corollary}

\theoremstyle{definition}
\newtheorem{definition}[theorem]{Definition}
\newtheorem{example}[theorem]{Example}

\newcommand{\R}{\mathbb{R}}
\newcommand{\N}{\mathbb{N}}

\newcommand{\id}{\mathrm{id}}
\newcommand{\Reg}{\mathrm{Reg}}
\newcommand{\Vcal}{\mathcal{V}}
\newcommand{\Ecal}{\mathcal{E}}
\newcommand{\Pcal}{\mathcal{P}}
\newcommand{\Scal}{\mathcal{S}}

\newcommand{\din}{D^{\mathrm{in}}}
\newcommand{\dout}{D^{\mathrm{out}}}

\makeatletter
\def\@email#1#2{%
 \endgroup
 \patchcmd{\titleblock@produce}
  {\frontmatter@RRAPformat}
  {\frontmatter@RRAPformat{\produce@RRAP{*#1\href{mailto:#2}{#2}}}\frontmatter@RRAPformat}
  {}{}
}%
\makeatother

\begin{document}


\title[Hamiltonian scattering on metric graphs]{Liouville-Preserving Hamiltonian Scattering on Finite Metric Graphs}

\author{Philip Hierhager}
\email{philip.hierhager@tum.de}
\affiliation{School of Computation, Information and Technology, Technical University of Munich, Munich, Germany}

\date{\today}

\begin{abstract}
A metric graph with a mechanical Hamiltonian on each edge does not, by itself,
define a deterministic classical motion through a branching vertex: conservation
of energy fixes only the outgoing speed, not the outgoing edge-end.  We study
the deterministic problem obtained after this missing vertex datum is supplied.
On each edge \(e\), with coordinate \(q\in[0,\ell_e]\), the Hamiltonian is
\(H_e(q,p)=p^2/2+V_e(q)\), where \(V\) is continuous on the graph and \(C^2\) on
every edge.  At each vertex we prescribe an energy-preserving Borel isomorphism
from incoming to outgoing nonzero boundary covectors.  The resulting phase space
is the measurable quotient that identifies each incoming boundary covector with
its prescribed outgoing one.  After excluding the finitely many energy levels
\(V(v)\), the edgewise Hamilton equations and the vertex laws concatenate to a
global one-parameter group of bimeasurable transformations.  The group preserves
energy and the quotient measure induced by the edgewise Liouville measures
\(dq\,dp\).  The proof uses no smooth symplectic structure on the quotient; the
invariance follows from ordinary edgewise Liouville invariance, a uniform
no-Zeno estimate on compact regular energy windows, and preservation of the
transverse Liouville flux \(r\,dr\) by the speedwise vertex permutations.  If the
vertex laws are compatible with momentum reversal, then the quotient dynamics is
reversible.  On regular energy surfaces satisfying the usual regular-value
condition, the induced time-parametrization measure is invariant as well.
\end{abstract}


\noindent\textit{The following article has been submitted to Journal of Mathematical Physics.}

\maketitle

\section{Introduction}

Hamiltonian mechanics on a smooth configuration manifold is locally determined
by the Hamiltonian vector field.  On a finite metric graph this remains true on
open edges but not at vertices.  Along an edge \(e\), the Hamiltonian
\[
        H_e(q,p)=\frac12p^2+V_e(q)
\]
gives the one-dimensional equations
\[
        \dot q=p,\qquad \dot p=-V_e'(q).
\]
At a branching vertex, however, the graph supplies only the finite set of
incident edge-ends.  Conservation of energy fixes the outgoing speed, whenever
\(E>V(v)\), by
\[
        |p_{\mathrm{out}}|=\sqrt{2(E-V(v))},
\]
but it does not select the outgoing edge-end.  Thus the metric graph and the
potential determine the motion along edges, but they do not by themselves
determine a deterministic continuation through a branching vertex.

For example, at a three-valent star vertex an incoming particle with energy
\(E>V(v)\) may be reflected back into the same edge, transmitted into one of
the other two edges, or routed according to an energy-dependent rule.  All of
these choices are compatible with conservation of energy, since the outgoing
speed is fixed by \(E\), but they define different deterministic mechanical
systems.  The missing datum is therefore a vertex continuation law.

This paper studies the deterministic classical problem obtained after these
vertex data have been prescribed.  At each vertex \(v\) we choose a Borel
isomorphism
\[
        S_v:\din_v\longrightarrow \dout_v
\]
from incoming to outgoing nonzero boundary covectors, and require it to preserve
the boundary energy.  The associated scattering phase space is the measurable
quotient obtained by imposing the identifications
\[
        \alpha\sim S_v\alpha,
        \qquad \alpha\in \din_v.
\]
Given these data, we show that the concatenated edge-and-vertex dynamics is a
global bimeasurable flow on the regular quotient phase space, preserving both
energy and the quotient Liouville measure.

The contribution is to isolate the minimal measurable structure needed for this
deterministic classical Hamiltonian construction.  The vertex behavior is
encoded by energy-preserving Borel isomorphisms of boundary covectors, with no
differentiability, symplecticity, or canonical selection rule assumed at the
vertices.  Under these hypotheses one nevertheless obtains a global
deterministic flow on a regular quotient phase space, together with energy
conservation and Liouville-measure preservation.  Thus the invariant measure is
a consequence of edgewise Hamiltonian invariance and boundary flux balance,
rather than of a smooth symplectic structure on the quotient.

There is one special case that will be excluded.  At a vertex \(v\),
the boundary energy relation is
\[
        E=\frac12p^2+V(v).
\]
If \(E>V(v)\), then the boundary momentum is nonzero and the trajectory has a
well-defined incoming or outgoing direction.  If \(E=V(v)\), however, a
trajectory can reach the vertex with \(p=0\).  Such a state is neither an
incoming nor an outgoing nonzero boundary covector, and hence it is not in the
domain of the prescribed scattering map.  We therefore work on the regular
sector obtained by excluding the finite set of vertex-potential energy levels
\[
        \{V(v):v\in\Vcal\}.
\]
On this sector every vertex transition is transverse, and the scattering law
determines the continuation uniquely.

In informal terms, the result says the following.  Once deterministic
energy-preserving scattering maps are supplied at the vertices, the ordinary
Hamiltonian motion on the edges can be concatenated through all vertex impacts
for all positive and negative times.  The resulting motion is a genuine
measurable flow on the regular quotient phase space and preserves both energy
and the quotient Liouville measure, even when the vertex laws are only Borel.

Throughout, ``Hamiltonian'' refers to the edgewise Hamiltonian dynamics and its
energy-preserving measurable concatenation at vertices.  We do not claim that
the quotient phase space is a smooth symplectic manifold or that the scattered
evolution is generated by a global smooth Hamiltonian vector field.

\subsection{Relation to previous work}

The construction in this paper belongs to the broad family of graph models in
which propagation along one-dimensional edges is coupled to additional data at
vertices.  This principle is standard in the theory of quantum graphs.  A
quantum graph is not merely a metric graph: one specifies differential
operators on the edges together with self-adjoint vertex conditions; see, for
example, Kuchment's survey and the monograph of Berkolaiko--Kuchment
\cite{Kuchment2004QuantumGraphsI,BerkolaikoKuchment2013}.  Thus the general
lesson that vertex behavior is extra structure is well established.  The
present paper uses this lesson in a classical deterministic setting, replacing
self-adjoint boundary conditions by energy-preserving Borel maps between
incoming and outgoing boundary covectors.

Time-continuous classical dynamics on graphs was studied by Barra and Gaspard
using Frobenius--Perron methods
\cite{BarraGaspard2001ClassicalDynamics,BarraGaspard2002Transport}, and
quantum-graph scattering with edge potentials was treated by Rueckriemen and
Smilansky in their trace-formula work
\cite{RueckriemenSmilansky2012EdgePotentials}.  The present paper is
complementary to these works: it gives a deterministic classical Hamiltonian
construction with prescribed energy-preserving vertex laws, a measurable
quotient phase space, and Liouville-measure preservation.

The scattering formulation of quantum graphs is also closely related in spirit.
Kottos--Smilansky introduced a periodic-orbit theory for quantum graphs in which
propagation along directed bonds is coupled to scattering at vertices, and in
which an associated classical graph dynamics underlies the trace formula and
spectral statistics \cite{KottosSmilansky1999}.  Gnutzmann--Smilansky's review
places this bond-scattering picture in the broader context of quantum chaos and
universal spectral statistics on graphs \cite{GnutzmannSmilansky2006}.  The
present construction differs from these works in three respects.  First, the
motion here is continuous in time along compact metric edges.  Second, the edge
speed is generated by the mechanical Hamiltonian
\[
        H_e(q,p)=\frac12p^2+V_e(q),
\]
with a possibly nonconstant potential.  Third, the vertex maps act on classical
boundary covectors and define a measurable quotient dynamics, rather than a
unitary quantum scattering matrix or a discrete directed-bond map.

There is also a large operator-theoretic literature describing quantum-graph
vertex data through boundary conditions and scattering matrices.  Kostrykin and
Schrader analyze general boundary conditions for quantum wires and the
corresponding on-shell scattering matrices \cite{KostrykinSchrader1999}.
Harmer relates self-adjoint extension theory on graphs to Hermitian symplectic
geometry \cite{Harmer2000Hermitian}.  Kurasov--Nowaczyk and
Caudrelier--Ragoucy further develop vertex-scattering descriptions and
computations of scattering matrices for quantum graphs
\cite{KurasovNowaczyk2010,CaudrelierRagoucy2010}.  These results concern
linear quantum dynamics and self-adjoint realizations of differential
operators.  In contrast, the present paper treats a real classical phase-space
problem: given deterministic vertex maps preserving the classical boundary
energy, construct the corresponding global scattered evolution and prove
measure preservation.

The closest classical analogues are impact Hamiltonian systems, billiard-type
systems, and Hamiltonian scattering maps.  Impact Hamiltonian systems describe
smooth Hamiltonian motion interrupted by impacts at boundaries; see, for
example, the treatment of impact Hamiltonian systems and polygonal billiards by
Becker--Elliott--Firester--Gonen Cohen--Pnueli--Rom-Kedar
\cite{BeckerEtAl2024ImpactHamiltonian}.  Buslaev--Pushnitski study scattering maps in
Hamiltonian mechanics and relate scattering symplectomorphisms to phase-volume
formulae \cite{BuslaevPushnitski2010Scattering}.  The present setting is different:
the configuration space is a singular one-dimensional complex rather than a
smooth domain with reflecting boundary or an asymptotic scattering problem, and
the vertex rule may transmit an incoming trajectory to a different edge-end
rather than reflect it specularly.

Finally, branching conditions also appear in quantum and stochastic dynamics on
graphs.  Exner--\v{S}eba studied free quantum motion on a branching graph using
self-adjoint extensions \cite{ExnerSeba1989}.  Kostrykin--Potthoff--Schrader
characterized Brownian motions on metric graphs by Wentzell-type boundary
conditions at the vertices \cite{KostrykinPotthoffSchrader2012}.  These
stochastic and quantum theories are not the subject of the present paper, but
they reinforce the same structural point: the path metric alone does not
determine the law of motion at a vertex.

Accordingly, the contribution is not a new principle of vertex scattering on
graphs, but a self-contained classical Hamiltonian construction in a setting
where the vertex laws are merely measurable.  For continuous-time mechanical
motion with Hamiltonians \(p^2/2+V_e(q)\) on compact metric edges, prescribed
energy-preserving Borel covector scattering gives a regular quotient phase
space, a no-Zeno global existence theorem, and invariance of the quotient
Liouville measure.

\subsection{Organization}

Section~\ref{sec:graph} fixes the graph and boundary notation.
Section~\ref{sec:scattering} defines scattering laws, their speed
representation, reversibility, and boundary flux measures.
Section~\ref{sec:phase} constructs the quotient phase space, the regular
sector, and the quotient Liouville measure.  Section~\ref{sec:dynamics}
constructs scattered trajectories and proves the no-Zeno estimate.
Section~\ref{sec:measure} proves measurability and Liouville invariance.
Section~\ref{sec:energy-surfaces} treats invariant measures on regular energy
surfaces.  Section~\ref{sec:discussion} records limitations and possible
extensions.

\subsection{Guide to notation}

The formal definitions begin in Sec.~\ref{sec:graph}.  For orientation, we use
the following symbols throughout.  An edge is denoted by \(e\), and
\(q_e\in[0,\ell_e]\) is its coordinate.  A vertex is denoted by \(v\), and an
edge-end incident to \(v\) is denoted by \(\eta=(e,a)\), where
\(a\in\{0,\ell_e\}\).  The sign \(\sigma(\eta)\) records which endpoint is used:
\(\sigma=-1\) at \(q_e=0\) and \(\sigma=+1\) at \(q_e=\ell_e\).  Thus
\(\sigma(\eta)p>0\) means that the covector points into the vertex, while
\(\sigma(\eta)p<0\) means that it points away from the vertex.

The sets \(\din_v\) and \(\dout_v\) are the incoming and outgoing nonzero
boundary covectors at \(v\).  The speed variable is always \(r=|p|>0\).  The map
\(S_v:\din_v\to\dout_v\) is the prescribed vertex scattering law.  In speed
coordinates it has the form \((\eta,r)\mapsto(F_v(\eta,r),r)\), where
\(F_v(\cdot,r)\) is a permutation of the finite set of edge-ends at fixed speed.
Tildes, as in \(\widetilde{\Pcal}\), denote objects before the boundary
identifications.  The quotient phase space is \(\Pcal_\Scal\), the quotient map
is \(\pi\), and the superscript \(\Reg\) denotes the regular sector obtained by
excluding the vertex-potential energy levels \(V(v)\).

\section{Metric graphs and boundary covectors}\label{sec:graph}

\begin{definition}[Finite metric graph]
A finite metric graph \(\Gamma\) consists of a finite combinatorial graph with
vertex set \(\Vcal\), edge set \(\Ecal\), and a length \(\ell_e\in(0,\infty)\)
for each edge \(e\).  Each edge is represented by an interval \([0,\ell_e]\),
and endpoints are glued according to the incidence relation.  The distance is
the induced path metric.  Multiple edges and loops are allowed; for a loop based
at a vertex, its two endpoint germs are counted separately.
\end{definition}

\begin{definition}[Edgewise regularity]
Let \(k\in\N\).  A function \(V:\Gamma\to\R\) is edgewise \(C^k\) if it is
continuous on \(\Gamma\) and, for every edge \(e\), the restriction \(V_e\) to
\([0,\ell_e]\) admits a \(C^k\) extension to an open interval containing
\([0,\ell_e]\).
\end{definition}

Throughout the paper \(V\) is assumed continuous on \(\Gamma\) and edgewise
\(C^2\). 
This regularity assumption is mainly technical.  We expect that it can be
weakened using Carathéodory or Filippov solution concepts, but only under
additional hypotheses ensuring that the edge dynamics is still deterministic,
energy preserving, and measure preserving up to the first boundary hit.
On the edge phase cylinder \([0,\ell_e]\times\R\) we use coordinates
\((q_e,p)\) and Hamiltonian
\begin{equation}\label{eq:edge-Hamiltonian}
        H_e(q_e,p)=\frac12p^2+V_e(q_e).
\end{equation}

Fix a vertex \(v\).  Let \(\Ecal_v^\partial\) be the finite set of edge-ends
incident to \(v\).  An element is written
\[
        \eta=(e,a),\qquad a\in\{0,\ell_e\},
\]
where the endpoint \(q_e=a\) is glued to \(v\).  Define
\[
        \sigma(\eta)=
        \begin{cases}
        -1, & a=0,\\
        +1, & a=\ell_e.
        \end{cases}
\]
Thus \(\sigma(\eta)p>0\) means that the Hamiltonian velocity \(\dot q_e=p\)
points toward the vertex through the edge-end \(\eta\), while
\(\sigma(\eta)p<0\) means that it points away from the vertex.

\begin{definition}[Incoming and outgoing boundary covectors]
For \(v\in\Vcal\), set
\[
        \din_v=
        \{(\eta,p):\eta\in\Ecal_v^\partial,\ p\ne0,\ \sigma(\eta)p>0\},
\]
and
\[
        \dout_v=
        \{(\eta,p):\eta\in\Ecal_v^\partial,\ p\ne0,\ \sigma(\eta)p<0\}.
\]
The boundary energy is
\[
        H_v^\partial(\eta,p)=\frac12p^2+V(v).
\]
\end{definition}

For \(r>0\), define
\[
        \alpha_{\eta,r}^{\mathrm{in}}=(\eta,\sigma(\eta)r),
        \qquad
        \alpha_{\eta,r}^{\mathrm{out}}=(\eta,-\sigma(\eta)r).
\]
Equivalently, define the speed charts
\[
        \iota_{\mathrm{in}}:
        \Ecal_v^\partial\times(0,\infty)\longrightarrow\din_v,
        \qquad
        \iota_{\mathrm{in}}(\eta,r)=\alpha_{\eta,r}^{\mathrm{in}},
\]
and
\[
        \iota_{\mathrm{out}}:
        \Ecal_v^\partial\times(0,\infty)\longrightarrow\dout_v,
        \qquad
        \iota_{\mathrm{out}}(\eta,r)=\alpha_{\eta,r}^{\mathrm{out}}.
\]
We equip \(\Ecal_v^\partial\) with the discrete sigma-algebra and
\(\Ecal_v^\partial\times(0,\infty)\) with the corresponding product Borel
structure.  With these structures, \(\iota_{\mathrm{in}}\) and
\(\iota_{\mathrm{out}}\) are Borel isomorphisms.

\section{Scattering laws}\label{sec:scattering}

\begin{definition}[Deterministic energy-preserving scattering law]\label{def:scattering-law}
A deterministic scattering law at \(v\) is a Borel isomorphism
\[
        S_v:\din_v\longrightarrow\dout_v.
\]
It is energy-preserving if
\[
        H_v^\partial(S_v\alpha)=H_v^\partial(\alpha),
        \qquad \alpha\in\din_v.
\]
A deterministic energy-preserving scattering structure is a family
\(\Scal=(S_v)_{v\in\Vcal}\) of such maps.
\end{definition}

\begin{proposition}[Speed representation]\label{prop:speed-representation}
Let \(S_v:\din_v\to\dout_v\) be energy-preserving.  Then there is a unique map
\[
        F_v:\Ecal_v^\partial\times(0,\infty)\longrightarrow \Ecal_v^\partial
\]
such that
\[
        S_v(\alpha_{\eta,r}^{\mathrm{in}})
        =\alpha_{F_v(\eta,r),r}^{\mathrm{out}}.
\]
Moreover, \(S_v\) is bijective if and only if \(F_v(\cdot,r)\) is a bijection
of \(\Ecal_v^\partial\) for every \(r>0\).
\end{proposition}

\begin{proof}
In speed coordinates, write
\[
        \iota_{\mathrm{out}}^{-1}S_v\iota_{\mathrm{in}}(\eta,r)=(\zeta,s).
\]
Energy preservation gives \(r^2=s^2\), hence \(s=r\), since \(r,s>0\).  This
proves the representation and uniqueness.  The bijectivity criterion follows
immediately because the coordinate representative has the form
\((\eta,r)\mapsto(F_v(\eta,r),r)\).
\end{proof}

\begin{lemma}[Borel criterion]\label{lem:borel-criterion}
Let \(S_v\) be an energy-preserving bijection with speed representation \(F_v\).
Then \(S_v\) is a Borel isomorphism if and only if, for all
\(\eta,\zeta\in\Ecal_v^\partial\), the set
\[
        \{r>0:F_v(\eta,r)=\zeta\}
\]
is Borel in \((0,\infty)\).
\end{lemma}

\begin{proof}
The product Borel structure is generated by sets \(\{\zeta\}\times B\).  The
coordinate representative \((\eta,r)\mapsto(F_v(\eta,r),r)\) has inverse
\((\zeta,r)\mapsto(F_v(\cdot,r)^{-1}\zeta,r)\).  Since \(\Ecal_v^\partial\) is
finite, measurability of both maps is exactly the stated Borel condition.
\end{proof}

Momentum reversal on the prequotient phase space is denoted by
\(\widetilde\rho\). On an edge cylinder,
\(
        \widetilde\rho(q,p)=(q,-p),
\)
and on boundary covectors
\(
        \widetilde\rho(\eta,p)=(\eta,-p).
\)
It exchanges \(\din_v\) and \(\dout_v\) and preserves \(H_v^\partial\).

\begin{definition}[Reversible scattering]\label{def:reversible-scattering}
A scattering law \(S_v:\din_v\to\dout_v\) is reversible if
\begin{equation}\label{eq:reversible-scattering}
        S_v^{-1}=\widetilde\rho\circ S_v\circ\widetilde\rho
        \qquad\text{as maps }\dout_v\to\din_v.
\end{equation}
\end{definition}

\begin{lemma}[Reversibility as a speedwise involution]\label{lem:reversibility-involution}
Let \(S_v\) be an energy-preserving bijection with speed representation \(F_v\).
Then \(S_v\) is reversible if and only if
\[
        F_v(F_v(\eta,r),r)=\eta
        \qquad (\eta\in\Ecal_v^\partial,\ r>0).
\]
Equivalently, for each fixed \(r>0\), \(F_v(\cdot,r)\) is an involution.
\end{lemma}

\begin{proof}
In incoming/outgoing speed coordinates, momentum reversal is the identity on
\(\Ecal_v^\partial\times(0,\infty)\) and only exchanges the two charts.  Thus
\(S_v^{-1}=\widetilde\rho S_v\widetilde\rho\) is equivalent to \(T_v^{-1}=T_v\), where
\(T_v(\eta,r)=(F_v(\eta,r),r)\).  This is equivalent to \(T_v^2=\id\), namely
to the displayed condition.
\end{proof}

\begin{definition}[Boundary flux measures]\label{def:boundary-flux}
The incoming and outgoing flux measures at \(v\) are
\[
        \mu_v^{\mathrm{in}}(A)
        =\sum_{\eta\in\Ecal_v^\partial}\int_0^\infty
        \mathbf 1_A(\alpha_{\eta,r}^{\mathrm{in}})r\,dr,
\]
and
\[
        \mu_v^{\mathrm{out}}(B)
        =\sum_{\eta\in\Ecal_v^\partial}\int_0^\infty
        \mathbf 1_B(\alpha_{\eta,r}^{\mathrm{out}})r\,dr.
\]
\end{definition}

\begin{lemma}[Flux preservation]\label{lem:flux-preservation}
Every deterministic energy-preserving scattering law satisfies
\[
        (S_v)_\#\mu_v^{\mathrm{in}}=\mu_v^{\mathrm{out}}.
\]
\end{lemma}

\begin{proof}
In speed coordinates the two flux measures are both counting measure on
\(\Ecal_v^\partial\) times \(r\,dr\).  The coordinate form of \(S_v\) is
\((\eta,r)\mapsto(F_v(\eta,r),r)\), and \(F_v(\cdot,r)\) is a permutation of the
finite set \(\Ecal_v^\partial\).  The product measure is therefore invariant.
\end{proof}

\begin{example}[Basic reversible laws]
Reflection is given by \(F_v(\eta,r)=\eta\).  Transmission through a two-valent
vertex is given by the transposition of its two edge-ends.  More generally, a
reversible deterministic law is, at each fixed speed, a product of fixed points
(reflections) and transpositions (transmissions).  A nonreflecting reversible
law can therefore exist only at even valence.
\end{example}

\begin{example}[A nonreversible cyclic law]
Let \(v\) be a vertex of valence three, with edge-ends
\(\eta_1,\eta_2,\eta_3\).  In speed coordinates define
\[
        F_v(\eta_1,r)=\eta_2,\qquad
        F_v(\eta_2,r)=\eta_3,\qquad
        F_v(\eta_3,r)=\eta_1
\]
for every \(r>0\).  This gives an energy-preserving Borel scattering law,
since for each fixed speed it is a permutation of the edge-ends.  It is not
reversible, because the permutation is a three-cycle rather than an involution.
The main construction still gives a deterministic measure-preserving scattered
flow, but the momentum-reversal symmetry of
Corollary~\ref{cor:global-reversibility} does not hold.
\end{example}

\section{The quotient scattering phase space}\label{sec:phase}

Let
\[
        \widetilde\Pcal=
        \coprod_{e\in\Ecal}([0,\ell_e]\times\R),
        \qquad
        \widetilde\Pcal^\circ=
        \coprod_{e\in\Ecal}((0,\ell_e)\times\R).
\]
Remove zero-momentum boundary points and set
\[
        \widetilde\Pcal^\times=
        \coprod_{e\in\Ecal}
        \bigl([0,\ell_e]\times\R\setminus(\{0,\ell_e\}\times\{0\})\bigr).
\]

We identify a boundary covector \((\eta,p)\), with \(\eta=(e,a)\), with the
boundary point \((a,p)\) in the \(e\)-component of
\(\widetilde\Pcal^\times\).  Under this convention, \(\din_v\) and
\(\dout_v\) are regarded as boundary subsets of \(\widetilde\Pcal^\times\).

Given \(\Scal=(S_v)_{v\in\Vcal}\), let \(\sim_\Scal\) be the equivalence
relation generated by
\[
        \alpha\sim_\Scal S_v\alpha,
        \qquad \alpha\in\din_v.
\]
All nontrivial equivalence classes are two-point classes
\(\{\alpha,S_v\alpha\}\).

\begin{definition}[Scattering phase space]\label{def:phase-space}
The scattering phase space is
\[
        \Pcal_\Scal=\widetilde\Pcal^\times/\sim_\Scal,
\]
with quotient map \(\pi:\widetilde\Pcal^\times\to\Pcal_\Scal\).  It is endowed
with the quotient measurable structure: \(A\subset\Pcal_\Scal\) is measurable
if and only if \(\pi^{-1}(A)\) is Borel in \(\widetilde\Pcal^\times\).
\end{definition}

\begin{lemma}[Finite quotient measurability]\label{lem:quotient-measurability}
If \(B\subset\widetilde\Pcal^\times\) is Borel, then its saturation
\(\pi^{-1}(\pi(B))\) is Borel.  Consequently \(\pi(B)\) is measurable in
\(\Pcal_\Scal\).
\end{lemma}

\begin{proof}
The saturation is obtained from \(B\) by adjoining, at finitely many vertices,
the images under \(S_v\) and \(S_v^{-1}\) of the corresponding boundary pieces
of \(B\).  These pieces are Borel and the scattering maps are Borel
isomorphisms.  A finite union of such sets is Borel.
\end{proof}

The edgewise Hamiltonian \(\widetilde H\) on \(\widetilde\Pcal^\times\) is
\(H_e\) on the \(e\)-component.  Energy preservation of the \(S_v\)'s implies
that \(\widetilde H\) is constant on equivalence classes.  Hence there is a
unique measurable function
\[
        H:\Pcal_\Scal\to\R,
        \qquad H\circ\pi=\widetilde H.
\]

\begin{definition}[Regular phase space]\label{def:regular-phase-space}
The regular scattering phase space is
\[
        \Pcal_\Scal^{\Reg}
        =\{z\in\Pcal_\Scal:H(z)\ne V(v)\text{ for every }v\in\Vcal\}.
\]
\end{definition}

At a boundary point over \(v\), the equality \(H_v^\partial(\eta,p)=V(v)\) is
equivalent to \(p=0\).  Thus every vertex impact in \(\Pcal_\Scal^{\Reg}\) has
nonzero boundary momentum.

Define the positive and negative boundary representatives
\[
        \widetilde\Pcal^+=\widetilde\Pcal^\circ\cup\coprod_{v\in\Vcal}\dout_v,
        \qquad
        \widetilde\Pcal^-=\widetilde\Pcal^\circ\cup\coprod_{v\in\Vcal}\din_v.
\]

\begin{lemma}[Outgoing and incoming quotient charts]\label{lem:quotient-charts}
The restrictions
\[
        \pi_+:=\pi|_{\widetilde\Pcal^+}:\widetilde\Pcal^+\to\Pcal_\Scal,
        \qquad
        \pi_-:=\pi|_{\widetilde\Pcal^-}:\widetilde\Pcal^-\to\Pcal_\Scal
\]
are bimeasurable bijections.
\end{lemma}

\begin{proof}
Each equivalence class contains exactly one point of \(\widetilde\Pcal^+\) and
exactly one point of \(\widetilde\Pcal^-\).  Hence the restrictions are
bijections.  Measurability follows from the quotient definition.  If
\(B\subset\widetilde\Pcal^\pm\) is Borel, then \(\pi(B)\) is measurable by
Lemma~\ref{lem:quotient-measurability}; hence the inverse maps are measurable.
\end{proof}

We write
\[
        j_+:=\pi_+^{-1},
        \qquad
        j_-:=\pi_-^{-1}.
\]

\begin{definition}[Quotient Liouville measure]\label{def:liouville-measure}
Let \(\lambda\) be the measure on \(\widetilde\Pcal^\times\) obtained by summing
Lebesgue measure \(dq_e\,dp\) over the finitely many edge cylinders.  The
quotient Liouville measure is
\[
        m=\pi_\#\lambda.
\]
Since \(\widetilde\Pcal^\times\setminus\widetilde\Pcal^\circ\) is contained in
a finite union of boundary lines, it is \(\lambda\)-null.  Thus
\[
        m(A)=\lambda\bigl(\pi^{-1}(A)\cap\widetilde\Pcal^\circ\bigr)
\]
for every measurable \(A\subset\Pcal_\Scal\).
\end{definition}

\begin{lemma}[Momentum reversal on the quotient]\label{lem:rho-descends}
If \(\Scal\) is reversible, then componentwise momentum reversal
\(\widetilde\rho\) descends to a measurable involution
\(\rho:\Pcal_\Scal\to\Pcal_\Scal\) satisfying
\[
        \rho\circ\pi=\pi\circ\widetilde\rho,
        \qquad H\circ\rho=H.
\]
It preserves \(\Pcal_\Scal^{\Reg}\).
\end{lemma}

\begin{proof}
It suffices to check that \(\widetilde\rho\) respects the equivalence relation.
Let \(\beta=S_v\alpha\).  By reversibility,
\(
        S_v(\widetilde\rho\beta)=\widetilde\rho\alpha,
\)
so \(\widetilde\rho\alpha\sim_\Scal\widetilde\rho\beta\).  Hence
\(
        \rho([x]):=[\widetilde\rho x]
\)
is well-defined and satisfies
\(
        \rho\circ\pi=\pi\circ\widetilde\rho .
\)
Measurability follows from the quotient measurable structure.  Since
\(\widetilde\rho^2=\id\), the descended map is an involution.  Finally,
\(H\circ\rho=H\) follows from the evenness of \(p^2/2\), and therefore
\(\rho\) preserves \(\Pcal_\Scal^{\Reg}\).
\end{proof}

\section{Scattered trajectories and no-Zeno estimates}\label{sec:dynamics}

For each edge choose once and for all a \(C^2\) extension \(\widetilde V_e\) to
an open interval \(I_e\supset[0,\ell_e]\), and let \(\phi_t^{(e)}\) be the
local flow of
\[
        \dot q=p,\qquad \dot p=-\widetilde V_e'(q).
\]
Inside the original edge cylinder, the solution is independent of the chosen
extension by uniqueness of ODE solutions.

\begin{lemma}[Edge flow]\label{lem:edge-flow}
On each edge cylinder, solutions of the edge equations are unique up to the
first time at which \(q\) reaches \(0\) or \(\ell_e\).  Along such a solution,
\(H_e\) and the measure \(dq_e\,dp\) are preserved.
\end{lemma}

\begin{proof}
Fix an edge \(e\), and use the chosen \(C^2\) extension
\(\widetilde V_e\) to an open interval containing \([0,\ell_e]\).  The extended
edge vector field is
\[
        X_e(q,p)=(p,-\widetilde V_e'(q)).
\]
It is \(C^1\), hence locally Lipschitz.  The standard Picard--Lindelöf
existence and uniqueness theorem for first-order systems therefore gives a
unique maximal solution through every initial condition in the extended edge
cylinder; see, for example, Sec.~10 of Ref.~\onlinecite{WalterODE}.  Restricting this solution to the connected time
interval on which \(q(t)\in[0,\ell_e]\) gives the maximal edgewise solution.
The restriction is independent of the chosen extension, since any two
extensions give the same vector field on \([0,\ell_e]\times\mathbb R\), and
ODE uniqueness identifies the corresponding solutions as long as they remain
inside the edge.

Energy conservation follows by differentiating along the solution:
\[
\begin{aligned}
        \frac{d}{dt}H_e(q(t),p(t))
        &=
        V_e'(q(t))\dot q(t)+p(t)\dot p(t)      \\
        &=
        V_e'(q(t))p(t)-p(t)V_e'(q(t))=0 .
\end{aligned}
\]
Finally, \(X_e\) is the Hamiltonian vector field of
\[
        H_e(q,p)=\frac12p^2+V_e(q)
\]
with respect to the canonical form \(dq_e\wedge dp\).  Hence the edge flow preserves \(dq_e\wedge dp\), and therefore the associated
Liouville measure \(dq_e\,dp\), by Liouville's theorem;
see for example Sec.~16 of Ref.~\onlinecite{Arnold1989} or equivalently, by preservation of the canonical
symplectic form under Hamiltonian flows,
Sec.~18.1 of Ref.~\onlinecite{CannasDaSilva2008}.
\end{proof}

\begin{lemma}[No-Zeno estimate]\label{lem:no-zeno}
Let \(K\subset\R\) be compact and suppose
\[
        K\cap\{V(v):v\in\Vcal\}=\varnothing.
\]
Then there exists \(\tau_K>0\) such that no scattered trajectory with energy in
\(K\) has two distinct vertex impacts separated by less than \(\tau_K\).
\end{lemma}

\begin{proof}
Set
\[
        \delta_K=\operatorname{dist}(K,\{V(v):v\in\Vcal\})>0.
\]
Since the graph is finite and \(V\) is continuous, choose \(r>0\) such that the
endpoint collars of length \(r\) are disjoint on every edge and
\[
        |V(q)-V(v)|\le \delta_K/4
\]
whenever \(q\) lies in the \(r\)-collar of an edge-end incident to \(v\).  Also
set
\[
        P_K=1+\sqrt{2(\max K-\min_\Gamma V)_+}.
\]
Every energy-\(E\in K\) trajectory satisfies \(|p|\le P_K\).

If such a trajectory lies in the \(r\)-collar of \(v\), then
\[
        |E-V(q)|\ge |E-V(v)|-|V(q)-V(v)|\ge 3\delta_K/4.
\]
Thus \(p^2=2(E-V(q))\) is nonzero there, and the sign of \(p\) cannot change
while the trajectory remains in that collar.  After an impact the trajectory
therefore must leave the collar before it can hit any vertex again.  It travels
edge-distance at least \(r\), while its speed is bounded by \(P_K\).  Hence the
next impact, if any, occurs after time at least \(r/P_K\).  Take
\(\tau_K=r/P_K\).
\end{proof}

\begin{lemma}[Local finiteness of impacts]\label{lem:impact-local-finiteness}
Let \(\gamma:\R\to\Pcal_\Scal^{\Reg}\) be an energy-preserving scattered
trajectory and let \(I_\gamma\) be its set of vertex-impact times.  If
\(E=H(\gamma(t))\), and \(\tau_E\) is the constant from
Lemma~\ref{lem:no-zeno} applied to \(K=\{E\}\), then distinct elements of
\(I_\gamma\) are separated by at least \(\tau_E\).  Consequently
\[
        \#(I_\gamma\cap[a,b])\le 1+\frac{b-a}{\tau_E}
\]
for every compact interval \([a,b]\).
\end{lemma}

\begin{proof}
Regularity gives \(E\notin\{V(v):v\in\Vcal\}\), so the no-Zeno lemma applies to
\(K=\{E\}\).  If \(t_1<\cdots<t_N\) are impacts in \([a,b]\), then
\(t_{j+1}-t_j\ge\tau_E\).  Summing these inequalities gives
\((N-1)\tau_E\le b-a\).
\end{proof}

\begin{definition}[Scattered trajectory]\label{def:scattered-trajectory}
A scattered trajectory through \(z\in\Pcal_\Scal^{\Reg}\) is a map
\(\gamma:\R\to\Pcal_\Scal^{\Reg}\) with \(\gamma(0)=z\) such that there is a
locally finite set \(I\subset\R\) with the following properties.  On every
component of \(\R\setminus I\), \(\gamma\) lifts to a single edge cylinder and
solves the edge Hamilton equations.  At each \(t_0\in I\), the left and right
one-sided lifts are an incoming covector \(\alpha\in\din_v\) and the outgoing
covector \(S_v\alpha\in\dout_v\), respectively, and
\[
        \gamma(t_0)=\pi(\alpha)=\pi(S_v\alpha).
\]
For negative-time construction the same condition is read with \(S_v^{-1}\).
\end{definition}

\begin{proposition}[Existence and uniqueness]\label{prop:trajectory-existence}
For every \(z\in\Pcal_\Scal^{\Reg}\) there is a unique scattered trajectory
\(\gamma_z:\R\to\Pcal_\Scal^{\Reg}\) through \(z\).
\end{proposition}

\begin{proof}
Use the outgoing representative
\(
        j_+(z)=\pi_+^{-1}(z),
\)
which is well-defined by Lemma~\ref{lem:quotient-charts}, to construct the
trajectory for \(t\ge0\).
If \(j_+(z)\) is an interior point, solve the edge ODE given by
Lemma~\ref{lem:edge-flow} until the first boundary hit.  If \(j_+(z)\) is an
outgoing boundary covector, the sign convention makes it point into the
corresponding edge for positive time, so the same construction applies.
At a boundary hit the regularity condition
ensures nonzero momentum; the covector is incoming and is replaced by its image
under \(S_v\).  Energy is preserved both along edges and at the scattering step.

Lemma~\ref{lem:no-zeno}, applied to \(K=\{H(z)\}\), prevents accumulation of impacts
in finite positive time, so the concatenation defines the forward trajectory for
all \(t\ge0\).  For \(t\le0\), use the incoming representative
\(j_-(z)=\pi_-^{-1}(z)\) and the inverse laws \(S_v^{-1}\).  If \(z\) is a
boundary class, the two one-sided representatives are \(\alpha\) and
\(S_v\alpha\), which have the same quotient value; hence the two one-sided
constructions agree at \(t=0\) in \(\Pcal_\Scal\).

Uniqueness follows by induction over impact intervals.  Between impacts it is
ODE uniqueness.  At an impact, the deterministic law fixes the outgoing
covector.  Local finiteness of impacts, Lemma~\ref{lem:impact-local-finiteness}, ensures
that the induction covers every compact time interval.
\end{proof}

Define
\[
        \Phi_tz:=\gamma_z(t),
        \qquad t\in\R,
        \quad z\in\Pcal_\Scal^{\Reg}.
\]
Uniqueness immediately gives
\begin{equation}\label{eq:flow-algebra}
        \Phi_{t+s}=\Phi_t\circ\Phi_s,
        \qquad
        \Phi_0=\id,
        \qquad
        \Phi_t^{-1}=\Phi_{-t}.
\end{equation}
Moreover,
\begin{equation}\label{eq:energy-conservation}
        H\circ\Phi_t=H,
        \qquad t\in\R,
\end{equation}
because the edge Hamiltonian is conserved along edge pieces and the scattering
laws preserve boundary energy.

\section{Measurability and Liouville invariance}\label{sec:measure}

The elementary algebra above does not yet prove that \(\Phi_t\) is measurable
or measure preserving.  We now give the technical part of the construction.

\begin{lemma}[Local flux form]\label{lem:local-flux-form}
Let \(\eta=(e,a)\) be an edge-end and let \(p_0\ne0\).  In a sufficiently small
flow box around \((a,p_0)\), the map
\[
        \Psi(s,\xi)=\phi_s^{(e)}(a,\xi)
\]
is a \(C^1\) diffeomorphism onto its image and satisfies
\[
        \Psi^*(dq_e\wedge dp)=\xi\,ds\wedge d\xi .
\]
Consequently the positive Liouville measure pulls back to
\[
        |\xi|\,ds\,d\xi ,
\]
and the transverse boundary measure is \(|\xi|\,d\xi\), equivalently \(r\,dr\)
in speed coordinates.
\end{lemma}

\begin{proof}
Work on a fixed \(C^2\) extension of \(V_e\) to an open interval containing
\([0,\ell_e]\).  The edge Hamiltonian vector field is
\[
        X_H(q,p)=(p,-V_e'(q)).
\]
Since \(p_0\ne0\), \(X_H(a,p_0)\) is transverse to the boundary section
\[
        \Sigma_\eta=\{(a,\xi):\xi \text{ near }p_0\}.
\]
Equivalently,
\[
        \Psi(s,\xi)=\phi_s^{(e)}(a,\xi)
\]
has, at \(s=0\),
\[
        \partial_s\Psi(0,\xi)=X_H(a,\xi)=(\xi,-V_e'(a)),
        \qquad
        \partial_\xi\Psi(0,\xi)=(0,1).
\]
Thus
\[
        \det
        \begin{pmatrix}
        \xi & 0\\
        -V_e'(a) & 1
        \end{pmatrix}
        =\xi .
\]
After shrinking the neighborhood of \(p_0\), this determinant is nonzero.
Hence, by the inverse function theorem, \(\Psi\) is a \(C^1\) diffeomorphism
from a small neighborhood of \((0,p_0)\) onto its image.

It remains to compute the pullback of the canonical two-form
\[
        \omega=dq_e\wedge dp .
\]
Hamiltonian edge flows preserve \(\omega\); see, for example, Sec.~18.1 of Ref.~\onlinecite{CannasDaSilva2008}. Therefore, for all \((s,\xi)\) in
the flow box,
\[
\begin{aligned}
        \omega_{\Psi(s,\xi)}
        \bigl(\partial_s\Psi(s,\xi),\partial_\xi\Psi(s,\xi)\bigr)
        &=
        \omega_{(a,\xi)}
        \bigl(X_H(a,\xi),(0,1)\bigr)  \\
        &=
        (dq_e\wedge dp)
        \bigl((\xi,-V_e'(a)),(0,1)\bigr)  \\
        &=\xi .
\end{aligned}
\]
Thus
\[
        \Psi^*(dq_e\wedge dp)=\xi\,ds\wedge d\xi .
\]

Taking absolute values gives the corresponding positive Liouville density
\[
        |\xi|\,ds\,d\xi .
\]
Thus the transverse measure induced on the boundary section is
\[
        |\xi|\,d\xi .
\]
On an incoming or outgoing branch the sign of \(\xi\) is fixed, and writing
\(r=|\xi|>0\) gives \(r\,dr\) in speed coordinates, up to orientation.  Since
the boundary flux measure is positive, the orientation sign is irrelevant.
\end{proof}

For a compact regular energy window \(K\), write
\[
        E_K:=H^{-1}(K)\cap\Pcal_\Scal^{\Reg}.
\]
For \(v\in\Vcal\), define the incoming and outgoing boundary energy windows
\[
        \din_{v,K}
        :=
        \{\alpha\in\din_v:H_v^\partial(\alpha)\in K\},
        \qquad
        \dout_{v,K}
        :=
        \{\beta\in\dout_v:H_v^\partial(\beta)\in K\}.
\]
For \(0<\Delta<\tau_K\), we use the following incoming and outgoing impact
tubes:
\[
        \Theta_v^{\mathrm{in}}(u,\alpha)
        =
        \pi\bigl(\phi_{-u}^{(e)}(\alpha)\bigr),
        \qquad
        \Theta_v^{\mathrm{out}}(w,\beta)
        =
        \pi\bigl(\phi_w^{(e')}(\beta)\bigr),
\]
where \(0<u,w<\Delta\), \(\alpha\in\din_{v,K}\),
\(\beta\in\dout_{v,K}\), and \(e,e'\) are the edge components determined by
the corresponding boundary covectors.  Thus
\(\Theta_v^{\mathrm{in}}(u,\alpha)\) is the point whose trajectory reaches
\(v\) with incoming covector \(\alpha\) after time \(u\), while
\(\Theta_v^{\mathrm{out}}(w,\beta)\) is the quotient point reached after time \(w\)
starting from \(v\) with outgoing covector \(\beta\).

\begin{lemma}[Short-time Borel structure]\label{lem:short-time-borel}
Let \(K\subset\mathbb R\) be compact and disjoint from
\(\{V(v):v\in\Vcal\}\), and let \(0<\Delta<\tau_K\).  On
\[
        E_K=H^{-1}(K)\cap\Pcal_\Scal^{\Reg},
\]
the no-impact set and the at-most-one-impact set for the time interval
\([0,\Delta]\) are measurable.  On each of these sets the short-time map
\(\Phi_\Delta\) is measurable.
\end{lemma}

\begin{proof}
Work in the outgoing quotient chart
\[
        j_+=\pi_+^{-1}:\Pcal_\Scal\to\widetilde\Pcal^+ .
\]
By Lemma~\ref{lem:quotient-charts}, this chart identifies
\(\Pcal_\Scal\) bimeasurably with \(\widetilde\Pcal^+\).  Hence
measurability may be checked in the concrete representative space
\(\widetilde\Pcal^+\), which is a finite disjoint union of open edge cylinders
together with outgoing boundary covectors.

Fix a vertex \(v\) and an incoming edge-end
\(\eta\in\Ecal_v^\partial\).  Set
\[
        R_{v,K}
        =
        \{r>0:\tfrac12 r^2+V(v)\in K\}.
\]
Since \(K\) is compact and disjoint from \(\{V(w):w\in\Vcal\}\),
\(R_{v,K}\) is compact in \((0,\infty)\).  Consider the incoming tube over this
boundary channel,
\[
        \Theta_{\eta}^{\mathrm{in}}(u,r)
        :=
        \Theta_v^{\mathrm{in}}
        \bigl(u,\alpha_{\eta,r}^{\mathrm{in}}\bigr),
        \qquad
        0<u<\Delta,\quad r\in R_{v,K}.
\]
Because \(\Delta<\tau_K\), the corresponding backward edge segment cannot hit
another vertex before time \(\Delta\); otherwise a trajectory with energy in
\(K\) would have two vertex impacts separated by less than \(\tau_K\).  Hence
\(\Theta_{\eta}^{\mathrm{in}}\) is well-defined on
\((0,\Delta)\times R_{v,K}\).

The map \(\Theta_{\eta}^{\mathrm{in}}\) is Borel, and before passing to the
quotient it is continuous.  It is also injective.  Indeed, if
\[
        \Theta_{\eta}^{\mathrm{in}}(u,r)
        =
        \Theta_{\eta}^{\mathrm{in}}(u',r'),
\]
then edge-flow uniqueness implies that the same trajectory reaches the same
boundary channel at times \(u\) and \(u'\).  If \(u\ne u'\), this would give two
vertex impacts less than \(\Delta<\tau_K\) apart.  Hence \(u=u'\), and then
uniqueness gives \(r=r'\).

Through the quotient chart \(j_+\), the range is a standard Borel space, and
the domain \((0,\Delta)\times R_{v,K}\) is also standard Borel.  By the
Lusin--Souslin theorem, the image
\[
        \mathcal I_{\eta}^{\circ}
        :=
        \Theta_{\eta}^{\mathrm{in}}
        \bigl((0,\Delta)\times R_{v,K}\bigr)
\]
is Borel, and the inverse coordinate map
\[
        (\Theta_{\eta}^{\mathrm{in}})^{-1}:
        \mathcal I_{\eta}^{\circ}
        \to
        (0,\Delta)\times R_{v,K}
\]
is Borel; see, for example, Theorem~15.1 of
Ref.~\onlinecite{Kechris2012}.

We also record the endpoint-impact slices.  Define
\[
        \mathcal I_{\eta}^{0}
        :=
        \Theta_{\eta}^{\mathrm{in}}
        \bigl(\{0\}\times R_{v,K}\bigr),
        \qquad
        \mathcal I_{\eta}^{\Delta}
        :=
        \Theta_{\eta}^{\mathrm{in}}
        \bigl(\{\Delta\}\times R_{v,K}\bigr),
\]
where the maps at \(u=0\) and \(u=\Delta\) are understood by continuous
extension of the corresponding edge-flow boxes.  These sets are Borel, since
\(R_{v,K}\) is Borel and the extended flow-box maps are Borel.

Taking the finite union over all vertices and incoming edge-ends gives the
Borel one-impact set
\[
        E_K^1
        =
        \bigcup_{v\in\Vcal}
        \bigcup_{\eta\in\Ecal_v^\partial}
        \left(
        \mathcal I_{\eta}^{\circ}
        \cup
        \mathcal I_{\eta}^{0}
        \cup
        \mathcal I_{\eta}^{\Delta}
        \right).
\]
By Lemma~\ref{lem:no-zeno}, no trajectory in \(E_K\) has more than one vertex
impact in \([0,\Delta]\).  Thus the no-impact set
\[
        E_K^0=E_K\setminus E_K^1
\]
is Borel.

On \(E_K^0\), the map \(\Phi_\Delta\) is ordinary edge Hamiltonian flow, hence
Borel.  On an open one-impact piece corresponding to a fixed vertex \(v\) and
incoming edge-end \(\eta\), write
\[
        z=\Theta_{\eta}^{\mathrm{in}}(u,r),
        \qquad
        0<u<\Delta,\quad r\in R_{v,K}.
\]
Then
\[
        \Phi_\Delta z
        =
        \Theta_v^{\mathrm{out}}
        \bigl(\Delta-u,S_v\alpha_{\eta,r}^{\mathrm{in}}\bigr).
\]
The coordinate inverse \(z\mapsto(u,r)\) is Borel on
\(\mathcal I_{\eta}^{\circ}\), the scattering map \(S_v\) is Borel, and the
outgoing tube map is Borel.  Hence \(\Phi_\Delta\) is Borel on each open
one-impact piece.

It remains only to check the endpoint slices.  On the slice
\(\mathcal I_{\eta}^{0}\), the impact occurs at the initial time.  Thus
\(\Phi_\Delta\) is given by first applying \(S_v\) to
\(\alpha_{\eta,r}^{\mathrm{in}}\) and then following the outgoing edge flow for
time \(\Delta\):
\[
        \Phi_\Delta
        \Theta_{\eta}^{\mathrm{in}}(0,r)
        =
        \Theta_v^{\mathrm{out}}
        \bigl(\Delta,S_v\alpha_{\eta,r}^{\mathrm{in}}\bigr).
\]
This is a composition of Borel maps.  On the slice
\(\mathcal I_{\eta}^{\Delta}\), the impact occurs at the terminal time.  The
map \(\Phi_\Delta\) sends the point to the quotient class of the corresponding
incoming covector, equivalently to the same class as its scattered outgoing
covector:
\[
        \Phi_\Delta
        \Theta_{\eta}^{\mathrm{in}}(\Delta,r)
        =
        \pi(\alpha_{\eta,r}^{\mathrm{in}})
        =
        \pi(S_v\alpha_{\eta,r}^{\mathrm{in}}).
\]
Again this is Borel in \(r\).  Since there are only finitely many vertices and
edge-ends, \(\Phi_\Delta\) is measurable on the one-impact set \(E_K^1\).

Combining the no-impact and one-impact parts, \(\Phi_\Delta\) is measurable on
\(E_K\).
\end{proof}

\begin{lemma}[One-impact Liouville invariance]\label{lem:one-impact-measure}
Under the hypotheses of Lemma~\ref{lem:short-time-borel}, the map
\(\Phi_\Delta\) preserves \(m\) on \(E_K\): for every measurable
\(A\subset\Pcal_\Scal^{\Reg}\),
\[
        m(\Phi_\Delta^{-1}A\cap E_K)=m(A\cap E_K).
\]
\end{lemma}

\begin{proof}
The sets of points whose trajectory has an impact exactly at time \(0\) or
\(\Delta\) are \(m\)-null.  In the tube coordinates introduced above, these
sets are contained in slices of the form
\[
        \{u=0\}\times B
        \qquad\text{or}\qquad
        \{u=\Delta\}\times B .
\]
By Lemma~\ref{lem:local-flux-form}, the Liouville measure in these coordinates
has density \(r\,du\,dr\), so such slices are null.  Since there are only
finitely many vertices and edge-ends, their union is null.  We discard this
null set.

On the no-impact part, \(\Phi_\Delta\) is ordinary edge Hamiltonian flow and
therefore preserves \(dq_e\,dp\) by Lemma~\ref{lem:edge-flow}.  It remains to
treat the one-impact part.

Fix a vertex \(v\).  Let \(\mathcal I_v^{\mathrm{in}}\) and
\(\mathcal I_v^{\mathrm{out}}\) be the incoming and outgoing one-impact strata
at \(v\), with endpoint slices removed.  By Lemma~\ref{lem:short-time-borel}
and the no-Zeno estimate, the tube coordinates give Borel isomorphisms onto
these strata.  Moreover, Lemma~\ref{lem:local-flux-form} gives
\[
        (\Theta_v^{\mathrm{in}})^*m=du\,d\mu_v^{\mathrm{in}},
        \qquad
        (\Theta_v^{\mathrm{out}})^*m=dw\,d\mu_v^{\mathrm{out}}.
\]
In these coordinates,
\[
        \Phi_\Delta\bigl(\Theta_v^{\mathrm{in}}(u,\alpha)\bigr)
        =
        \Theta_v^{\mathrm{out}}\bigl(\Delta-u,S_v\alpha\bigr).
\]

Let \(f\ge0\) be measurable on \(\mathcal I_v^{\mathrm{out}}\).  Then
\[
\begin{aligned}
        \int_{\mathcal I_v^{\mathrm{in}}}
        f(\Phi_\Delta z)\,dm(z)
        &=
        \int_{\din_{v,K}}\int_0^\Delta
        f\!\left(
        \Theta_v^{\mathrm{out}}(\Delta-u,S_v\alpha)
        \right)
        du\,d\mu_v^{\mathrm{in}}(\alpha)       \\
        &=
        \int_{\dout_{v,K}}\int_0^\Delta
        f\!\left(
        \Theta_v^{\mathrm{out}}(w,\beta)
        \right)
        dw\,d\mu_v^{\mathrm{out}}(\beta)       \\
        &=
        \int_{\mathcal I_v^{\mathrm{out}}} f(z)\,dm(z).
\end{aligned}
\]
The second equality uses \(w=\Delta-u\) and
Lemma~\ref{lem:flux-preservation}.  Applying this identity to indicator
functions and summing over the finitely many vertices gives measure
preservation on the one-impact part.  Combining this with the no-impact part
and the discarded null endpoint slices yields
\[
        m(\Phi_\Delta^{-1}A\cap E_K)=m(A\cap E_K)
\]
for every measurable \(A\subset\Pcal_\Scal^{\Reg}\).
\end{proof}

\begin{lemma}[Measurability of the scattered evolution]\label{lem:measurability}
For every \(t\in\R\), the map
\[
        \Phi_t:\Pcal_\Scal^{\Reg}\to\Pcal_\Scal^{\Reg}
\]
is measurable.  Hence, by \eqref{eq:flow-algebra}, it is bimeasurable.
\end{lemma}

\begin{proof}
Let \(C=\{V(v):v\in\Vcal\}\) and set
\[
        K_n=\{E\in[-n,n]:\operatorname{dist}(E,C)\ge 1/n\}.
\]
Then \(\R\setminus C=\bigcup_n K_n\), with each \(K_n\) compact and disjoint
from \(C\).  The measurable invariant sets
\[
        E_n=H^{-1}(K_n)\cap\Pcal_\Scal^{\Reg}
\]
exhaust \(\Pcal_\Scal^{\Reg}\).  For fixed \(t>0\), partition \([0,t]\) into
subintervals of length smaller than \(\tau_{K_n}\).  On \(E_n\), each short-time
factor is measurable by Lemma~\ref{lem:short-time-borel}; hence their finite
composition is measurable.  Since the \(E_n\)'s exhaust the space, \(\Phi_t\) is
measurable.  The case \(t<0\) is identical, using the incoming chart
\(j_-=\pi_-^{-1}\) from Lemma~\ref{lem:quotient-charts} and the Borel maps
\(S_v^{-1}\), which exist by Definition~\ref{def:scattering-law}.
\end{proof}

\begin{theorem}[Global deterministic scattered evolution]\label{thm:main}
Let \(\Gamma\) be a finite metric graph, let \(V\) be continuous and edgewise
\(C^2\), and let \(\Scal=(S_v)_{v\in\Vcal}\) be deterministic
energy-preserving scattering laws.  Then the maps \((\Phi_t)_{t\in\R}\) form a
one-parameter group of bimeasurable transformations of \(\Pcal_\Scal^{\Reg}\).
For all \(t\in\R\),
\[
        H\circ\Phi_t=H,
\]
and
\[
        m(\Phi_t^{-1}A)=m(A)
\]
for every measurable \(A\subset\Pcal_\Scal^{\Reg}\), where \(m\) is restricted
to \(\Pcal_\Scal^{\Reg}\).
\end{theorem}

\begin{proof}
The group property and energy conservation were proved in
\eqref{eq:flow-algebra} and \eqref{eq:energy-conservation}; measurability is
Lemma~\ref{lem:measurability}.  For measure preservation it suffices to consider \(t\ge0\), since
\(\Phi_t^{-1}=\Phi_{-t}\).  Let
\(C=\{V(v):v\in\Vcal\}\) and use the compact exhaustion
\[
        K_n=\{E\in[-n,n]:\operatorname{dist}(E,C)\ge 1/n\}.
\]
The invariant measurable sets
\[
        E_n=H^{-1}(K_n)\cap\Pcal_\Scal^{\Reg}
\]
increase to \(\Pcal_\Scal^{\Reg}\).  On each \(E_n\), partition the time
interval into finitely many subintervals of length less than \(\tau_{K_n}\).
Lemma~\ref{lem:one-impact-measure} gives measure preservation for each
short-time factor, and the invariance of \(E_n\) under the flow allows these
identities to be iterated.  Hence \(\Phi_t\) preserves \(m\) on every \(E_n\).

For a measurable \(A\subset\Pcal_\Scal^{\Reg}\),
\[
        A\cap E_n\uparrow A,
        \qquad
        \Phi_t^{-1}(A)\cap E_n
        =
        \Phi_t^{-1}(A\cap E_n)\cap E_n
        \uparrow \Phi_t^{-1}(A).
\]
By monotone convergence of measures,
\[
        m(\Phi_t^{-1}A)
        =
        \lim_{n\to\infty}m(\Phi_t^{-1}(A\cap E_n)\cap E_n)
        =
        \lim_{n\to\infty}m(A\cap E_n)
        =
        m(A).
\]
Thus \(\Phi_t\) preserves \(m\) on the whole regular phase space.
\end{proof}

\begin{corollary}[Reversibility of the global evolution]
\label{cor:global-reversibility}
Assume that the scattering structure \(\Scal=(S_v)_{v\in\Vcal}\) is reversible,
that is,
\[
        S_v^{-1}=\widetilde\rho\circ S_v\circ\widetilde\rho
        \qquad\text{as maps }\dout_v\to\din_v
\]
for every \(v\in\Vcal\).  Let
\[
        \rho:\Pcal_\Scal\to\Pcal_\Scal
\]
denote the quotient momentum-reversal involution from
Lemma~\ref{lem:rho-descends}.  Then \(\rho\) preserves
\(\Pcal_\Scal^{\Reg}\), and the scattered evolution satisfies
\[
        \rho\circ\Phi_t=\Phi_{-t}\circ\rho,
        \qquad t\in\mathbb R .
\]
\end{corollary}

\begin{proof}
By Lemma~\ref{lem:rho-descends}, componentwise momentum reversal descends to a
measurable involution on the quotient phase space, still denoted by \(\rho\),
and satisfies
\[
        H\circ\rho=H.
\]
Since regularity is defined by excluding the energy levels \(V(v)\), it follows
that \(\rho\) maps \(\Pcal_\Scal^{\Reg}\) onto itself.

Fix \(z_0\in\Pcal_\Scal^{\Reg}\), and let
\[
        z(t):=\Phi_t z_0
\]
be the unique scattered trajectory through \(z_0\), given by
Proposition~\ref{prop:trajectory-existence}.  Define
\[
        \widetilde z(t):=\rho z(-t).
\]
We show that \(\widetilde z\) is the scattered trajectory through
\(\rho z_0\).

First consider a time interval on which \(z(t)\) has a lift
\[
        (q(t),p(t))
\]
to a single edge cylinder and satisfies
\[
        \dot q(t)=p(t),
        \qquad
        \dot p(t)=-V_e'(q(t)).
\]
Then \(\widetilde z(t)\) has the lifted representative
\[
        (\widetilde q(t),\widetilde p(t))
        :=
        (q(-t),-p(-t)).
\]
Differentiating gives
\[
        \dot{\widetilde q}(t)
        =
        -\dot q(-t)
        =
        -p(-t)
        =
        \widetilde p(t),
\]
and
\[
        \dot{\widetilde p}(t)
        =
        \dot p(-t)
        =
        -V_e'(q(-t))
        =
        -V_e'(\widetilde q(t)).
\]
Thus \(\widetilde z\) satisfies the same edge Hamilton equations on every
edgewise interval.

It remains to check the vertex rule.  Suppose \(t_0\) is an impact time for
\(z\) at a vertex \(v\).  Let the incoming and outgoing one-sided boundary
covectors of \(z\) at \(t_0\) be
\[
        \alpha\in\din_v,
        \qquad
        S_v\alpha\in\dout_v.
\]
Thus, in the quotient,
\[
        z(t_0)=\pi(\alpha)=\pi(S_v\alpha).
\]
For the reversed curve \(\widetilde z(t)=\rho z(-t)\), the corresponding impact
time is \(-t_0\).  Immediately before \(-t_0\), the reversed curve has boundary
covector
\[
        \widetilde\rho(S_v\alpha)\in\din_v,
\]
and immediately after \(-t_0\), it has boundary covector
\[
        \widetilde\rho\alpha\in\dout_v.
\]
The reversibility condition gives
\[
        S_v(\widetilde\rho S_v\alpha)=\widetilde\rho\alpha .
\]
Therefore the one-sided boundary covectors of \(\widetilde z\) at \(-t_0\)
satisfy precisely the prescribed forward-time scattering rule.

Hence \(\widetilde z\) is a scattered trajectory through
\[
        \widetilde z(0)=\rho z(0)=\rho z_0.
\]
By uniqueness of scattered trajectories,
Proposition~\ref{prop:trajectory-existence}, we have
\[
        \widetilde z(t)=\Phi_t(\rho z_0)
        \qquad\text{for all }t\in\mathbb R.
\]
Using the definition of \(\widetilde z\), this gives
\[
        \rho\Phi_{-t}z_0=\Phi_t\rho z_0 .
\]
Replacing \(t\) by \(-t\), we obtain
\[
        \rho\Phi_t z_0=\Phi_{-t}\rho z_0 .
\]
Since \(z_0\in\Pcal_\Scal^{\Reg}\) was arbitrary,
\[
        \rho\circ\Phi_t=\Phi_{-t}\circ\rho,
        \qquad t\in\mathbb R .
\]
\end{proof}

\section{Regular energy surfaces}\label{sec:energy-surfaces}

For \(E\notin\{V(v):v\in\Vcal\}\), the level
\(H^{-1}(E)\cap\Pcal_\Scal^{\Reg}\) is invariant by
Theorem~\ref{thm:main}.
The following proposition records the usual time-parametrization measure on a
one-dimensional regular Hamiltonian energy curve.

\begin{proposition}[Invariant time measure on a regular energy surface]\label{prop:energy-surface-measure}
Let \(E_0\notin\{V(v):v\in\Vcal\}\), and assume that \(E_0\) is a regular value
of \(V_e:(0,\ell_e)\to\R\) for every edge \(e\).  Set
\[
        U_e(E_0)=\{q\in(0,\ell_e):E_0>V_e(q)\},
        \qquad
        \Sigma_{E_0}=H^{-1}(E_0)\cap\Pcal_\Scal^{\Reg}.
\]
Define \(\nu_{E_0}\) as the pushforward measure characterized by
\[
\begin{aligned}
        \int_{\Sigma_{E_0}} f\,d\nu_{E_0}
        :=\sum_{e\in\Ecal}\sum_{\varepsilon=\pm1}
        \int_{U_e(E_0)}
        f\!\left(\pi(e,q,\varepsilon\sqrt{2(E_0-V_e(q))})\right)
        \frac{dq}{\sqrt{2(E_0-V_e(q))}}
\end{aligned}
\]
for every nonnegative measurable \(f\).  Equivalently, \(\nu_{E_0}\) is the
pushforward of the branch measures
\[
        \frac{dq}{\sqrt{2(E_0-V_e(q))}}
\]
under the maps
\[
        q\mapsto
        \pi(e,q,\varepsilon\sqrt{2(E_0-V_e(q))}),
        \qquad q\in U_e(E_0),\quad \varepsilon=\pm1.
\]
With this convention, \(\nu_{E_0}\) is a measure on the whole quotient energy
surface \(\Sigma_{E_0}\).
The complement of the parametrized open branches
consists only of interior turning points and boundary equivalence classes, and
this complement has \(\nu_{E_0}\)-measure zero.  Then
\(\nu_{E_0}\) is invariant under
\(\Phi_t|_{\Sigma_{E_0}}\).  If \(\Sigma_{E_0}\ne\varnothing\), then
\(\nu_{E_0}(\Sigma_{E_0})<\infty\), and its normalization is an invariant
probability measure.
\end{proposition}

\begin{proof}
On a branch of the regular energy level,
\[
        p=\varepsilon\sqrt{2(E_0-V_e(q))},
        \qquad \varepsilon=\pm1,
\]
Hamilton's equation \(\dot q=p\) gives
\[
        dt=\frac{dq}{p}
\]
with orientation, and hence the positive time-parametrization measure is
\[
        |dt|=\frac{dq}{|p|}
        =
        \frac{dq}{\sqrt{2(E_0-V_e(q))}} .
\]
This is the standard time measure on a one-dimensional Hamiltonian energy curve.

We first check finiteness.  Since \(E_0\) is a regular value of
\(V_e:(0,\ell_e)\to\mathbb R\), every interior zero \(q_0\) of
\(E_0-V_e\) satisfies \(V_e'(q_0)\ne0\).  Hence \(E_0-V_e(q)\) vanishes
linearly at \(q_0\), and
\[
        \frac{dq}{\sqrt{E_0-V_e(q)}}
\]
has only an integrable square-root singularity there.  These zeros are isolated
and cannot accumulate in the interior of the edge.  They also cannot accumulate
at an endpoint, since \(E_0\ne V(v)\) at every vertex and \(V\) is continuous.
Thus each compact edge contains only finitely many turning points, and the
total measure is finite because the graph has finitely many edges.

Interior turning points, namely points with \(E_0=V_e(q)\) and \(p=0\),
belong to the energy surface but are not included in the branch parametrization
over \(U_e(E_0)\).  This causes no ambiguity: the preceding estimate shows that
the time measure has finite mass in a punctured neighborhood of each such point,
and adjoining the turning point itself changes the measure by zero.  Thus
\(\nu_{E_0}\) is a measure on the quotient energy surface whose possible
interior zero-momentum points are \(\nu_{E_0}\)-null.

The definition of \(\nu_{E_0}\) is compatible with the quotient.  The only
nontrivial quotient identifications on \(\Sigma_{E_0}\) occur at boundary
covectors, while the formula defining \(\nu_{E_0}\) integrates over open edge
sets \(U_e(E_0)\).  Thus the boundary equivalence classes have
\(\nu_{E_0}\)-measure zero, and the use of
\[
        \pi(e,q,\varepsilon\sqrt{2(E_0-V_e(q))})
\]
defines a measure on the quotient energy surface.

Away from vertices and interior turning points, the flow is translation in the
time coordinate on each energy branch and therefore preserves \(|dt|\).
At an interior turning point \(q_0\), the assumption \(V_e'(q_0)\ne0\) implies
that the two momentum branches join into a regular one-dimensional energy
curve.  The coordinate \(t\), defined by \(dt=dq/p\) with orientation, extends
through the turning point with finite total variation.  Since the turning point
itself is \(\nu_{E_0}\)-null, passage through it preserves the
time-parametrization measure.

At a vertex \(v\), the scattering map preserves the
boundary speed
\[
        r=\sqrt{2(E_0-V(v))}
\]
and permutes the finitely many boundary channels.  Since the vertex itself has
zero \(\nu_{E_0}\)-measure, the transition across the vertex does not change
the time-parametrization measure.

By Lemma~\ref{lem:impact-local-finiteness},
a bounded time interval contains only finitely many vertex transitions.  Hence
the flow on \(\Sigma_{E_0}\) is obtained by concatenating finitely many
measure-preserving edge translations and speed-preserving vertex permutations
on every bounded time interval.  Therefore
\[
        (\Phi_t)_\#\nu_{E_0}=\nu_{E_0}
\]
for every \(t\in\mathbb R\).  Normalizing by the finite total mass gives an
invariant probability measure when \(\Sigma_{E_0}\ne\varnothing\).
\end{proof}

\begin{example}[Reflecting interval]
Let \(\Gamma=[0,L]\) and put reflecting laws at both endpoints.  The theorem
recovers the usual one-dimensional Hamiltonian motion with elastic reflection,
restricted to energies different from the endpoint values of \(V\).  On a
regular energy level, Proposition~\ref{prop:energy-surface-measure} gives the
standard invariant density \(dq/\sqrt{2(E-V(q))}\) on each momentum branch.
\end{example}

\section{Conclusion and limitations}\label{sec:discussion}

The construction is conditional in the physically relevant sense that the
metric graph and the potential determine only the edgewise Hamiltonian motion.
At a branching vertex, conservation of energy determines the outgoing speed but
not the outgoing edge-end.  The scattering laws are therefore genuine mechanical
data, analogous in role to boundary conditions in quantum or stochastic graph
models.  Different scattering laws on the same metric graph can produce
different deterministic classical systems, even though they share the same
edge Hamiltonian and the same conserved energy.  This is precisely why the
vertex maps are part of the mathematical model rather than a choice made inside
the proof.

The phase space used here is a measurable quotient of the disjoint union of the
edgewise Hamiltonian phase spaces \(T^*[0,\ell_e]\cong[0,\ell_e]\times\mathbb R\).
The results establish a bimeasurable, measure-preserving scattered evolution.
They do not assert that the quotient is a smooth symplectic manifold or that the
scattered evolution is generated by a global Hamiltonian vector field.

The exclusion of the levels \(V(v)\) is a global regularity condition ensuring
that every vertex impact has nonzero momentum.  A more refined theory could keep
states with such energies when their trajectories never encounter the relevant
vertex, but that refinement is not pursued here.

Finiteness of the graph is used in three places: the set of vertex energy values
is finite, the number of edge-ends is finite, and the no-Zeno estimate has a
uniform collar size.  Infinite or locally finite graphs would require explicit
compactness and non-accumulation assumptions.  Further extensions include
zero-momentum vertex rules, stochastic or multivalued scattering, smoother
scattering laws leading to symplectic structures on regular strata, and
infinite graphs under suitable no-Zeno hypotheses.

\section*{Author Declarations}

\subsection*{Conflict of Interest}
The author has no conflicts to disclose.

\subsection*{Author Contributions}
Philip Hierhager: Conceptualization, Methodology, Formal analysis,
Investigation, Writing -- original draft, Writing -- review and editing.

\subsection*{Data Availability}
Data sharing is not applicable to this article as no new data were created or analyzed in this study. 

\subsection*{Use of AI Tools}
The author used OpenAI ChatGPT for language refinement, organizational
suggestions, and exploratory assistance during the preparation of this
manuscript. In particular, the tool was used to suggest possible proof
strategies, identify potential gaps or notational ambiguities, and improve the
presentation of definitions and proofs. 
All mathematical claims, proofs,
citations, and conclusions were independently reviewed, revised, and verified by
the author, who takes full responsibility for the final content of the
manuscript.

\bibliography{bib}

\end{document}